\documentclass[11pt,reqno]{amsart}

\usepackage{amsthm,amsfonts,amsmath,amssymb,bookmark,appendix}
\usepackage{amsmath,amsfonts,amsthm,amssymb,amsxtra}
\usepackage{amsxtra, amssymb, mathrsfs}
\usepackage[usenames,dvipsnames]{color}

\newtheorem{thm}{Theorem}[section]
\newtheorem{prop}[thm]{Proposition}

\newtheorem{lemma}[thm]{Lemma}
\newtheorem{corollary}[thm]{\bf Corollary}
\newtheorem{assumption}[thm]{Assumption}

\theoremstyle{definition}

\theoremstyle{remark}
\newtheorem{remark}[thm]{\bf Remark}

\numberwithin{equation}{section}

\makeatletter
\renewcommand*\env@cases[1][1.6]{%
  \let\@ifnextchar\new@ifnextchar
  \left\lbrace
  \def\arraystretch{#1}%
  \array{@{}l@{\quad}l@{}}%
}
\makeatother

\newcommand{\norm}[1]{\mbox{$\left\| #1 \right\|$}}
\newcommand{\sprod}[2]{\mbox{$\left\langle {\,#1},{\, #2} \right\rangle$}}

\newcommand{\R}{\mathbb{R}}

\renewcommand{\Re}{\text{Re}}

\newcommand{\df}{\ \mathrm{d}}
\newcommand{\menge}[2]{\left\{ \, {#1}  \ \big{|} \ {#2} \, \right\}}
\newcommand{\iii}{\mathrm{i}}

\renewcommand{\Re}{\mathrm{Re}}

\renewcommand{\exp}[1]{\mathrm{e}^{#1}}

\newcommand{\trace}[1]{\mathrm{Tr}{\left(#1\right)}}
\newcommand{\dist}[2]{\mathrm{dist}\left({#1},{#2}\right)}
 \newcommand{\lam}{\lambda_{j}}

\numberwithin{equation}{section}
 
\begin{document}

\title{Spectral estimates for the Heisenberg Laplacian on cylinders}


\author{Hynek Kova\v r\'{\i}k}
\address{Hynek Kova\v r\'{\i}k, DICATAM, Sezione di Matematica, Universit\`a degli studi di Brescia,
Via Branze, 38 - 25123  Brescia, Italy}
\email{hynek.kovarik@ing.unibs.it}

\author{Bartosch Ruszkowski}
\address{Bartosch Ruszkowski, Institute of Analysis, Dynamics
and Modeling, Universit\"at Stuttgart, 
PF 80 11 40, D-70569  Stuttgart, Germany}
\email{Bartosch.Ruszkowski@mathematik.uni-stuttgart.de}

\author{Timo Weidl}
\address{Timo Weidl, Institute of Analysis, Dynamics
and Modeling, Universit\"at Stuttgart, 
PF 80 11 40, D-70569  Stuttgart, Germany}
\email{weidl@mathematik.uni-stuttgart.de}

\begin{abstract} We study Riesz means of  eigenvalues of
the Heisenberg Laplacian with Dirichlet boundary conditions on a cylinder in dimension three. We obtain an inequality with a sharp leading term and an additional lower order term.
\end{abstract}
\maketitle

\section{\bf Introduction}

Let $\Omega\subset \R^3$ be an open bounded domain. We consider the Heisenberg Laplacian on $\Omega$ with Dirichlet boundary condition, formally given by
\begin{align*}
\mathrm{A}(\Omega)\ := \ -X_{1}^{2}-X^{2}_{2} \quad \text{with} \quad X_1  \ :=  \ \partial_{x_1}+\frac{x_2}{2}\partial_{x_3}, \quad X_2 \ := \ \partial_{x_2}-\frac{x_1}{2}\partial_{x_3}\, .
\end{align*}
This operator is associated with the closure of the quadratic form
\begin{align}\label{einseins}
\mathrm{ {a}}[u]& \ :=  \ \int_{\Omega} \left|X_{1}u(x)\right|^{2}+\left|X_{2}u(x) \right|^{2} \df x,
\end{align}
initially defined on $C_{0}^{\infty}(\Omega)$.  It is known, see e.g.~\cite{follandsub},  \cite{hanson}, \cite{melas-type}, that 
$\mathrm{A}(\Omega)$ has purely discrete spectrum.
We denote by $(\lambda_{k}(\Omega))_{k\in \mathbb{N}}$ the non-decreasing unbounded sequence of the eigenvalues of  $\mathrm{A}(\Omega)$, where we repeat entries according to their finite multiplicities. We are interested in uniform upper bounds on the quantity 
\begin{equation*}
\trace{\mathrm{A}(\Omega)-\lambda}_-   =  \ \sum_{k=1}^{\infty} \, \left(\lambda_{k}(\Omega) -\lambda\right)_{-}  . 
\end{equation*}
In \cite{hanson} Hansson and Laptev proved the following Berezin-type inequality for $\mathrm{A}(\Omega)$: 
\begin{align}\label{result}
\trace{\mathrm{A}(\Omega)-\lambda}_{-} \ \leq \  \frac{|\Omega|}{96}\ \lambda^3 \qquad \forall\ \lambda>0.
\end{align}
It is also shown in \cite{hanson} that 
\begin{align} \label{asymp}
\sum_{k=1}^{\infty}\,  (\lambda-\lambda_{k}(\Omega))_{+}  \ = \ \frac{|\Omega|}{96}\,  \lambda^3 + o(\lambda^3) \qquad \text{as} \quad \lambda\to +\infty, 
\end{align}
which implies that the constant $\frac{1}{96}$ on the right hand side of \eqref{result} is sharp.

Nevertheless, the authors of the present paper proved, see \cite{melas-type}, that inequality \eqref{result} can be improved in the following sense; for a any bounded domain $\Omega\subset \R^3$, there exists a constant $C(\Omega)>0$ such that for any $\lambda\geq0$ it holds
\begin{align} \label{eq-krw1}
\trace{\mathrm{A}(\Omega) -\lambda}_{-} \ \leq  \ \max \left\{0,\  \frac{|\Omega|}{96} \ \lambda^{3}- C(\Omega)\, \lambda^{2} \right\} .
\end{align}
In other words, a negative remainder term of a lower order can be added to the right hand side of \eqref{result} without violating the inequality. 

\smallskip

\noindent In this paper we will prove that the order of the remainder term in \eqref{eq-krw1} can be further improved if we consider cylindrical domains of the type $\Omega = \omega\times (a,b)$, where $\omega \subset{\mathbb{R}^{2}}$ is open and bounded, and $a,b \in \R$ are such that $a<b$. In particular for cylinders wit convex cross-section $\omega$ our main result, Theorem \ref{dreieins}, implies that
\begin{align} \label{eq-krw2}
\trace{\mathrm{A}(\Omega)-\lambda }_{-}\ \leq\  \max\left\{0,\frac{|\Omega|}{96}\,  \lambda^{3}
-   \frac{\lambda^{2+\frac{1}{4}}}{2^7\cdot 3^{5/2} }\,  \frac{|{\Omega}|}{\mathrm{R}{(\omega)^{3/2}}}  \right\}\, ,
\end{align}
where $\mathrm{R}(\omega)$ is the in-radius of $\omega$, see Corollary \ref{cor1}. 
The proof of \eqref{eq-krw2} is based on the unitary equivalence of $\mathrm{A}(\Omega)$ 
to the two-dimensional Laplacian with constant magnetic field. To estimate the remainder term we use a boundary estimate for the magnetic Laplacian based on an application of a Hardy inequality in the spirit of \cite{davisss}, see Proposition \ref{prop}.

\section{\bf Notation and main results}\label{main}

\noindent As for the cross-section $\omega$, throughout the paper we will suppose that the following condition is satisfied.

\begin{assumption}\label{assumption}
The open domain ${\omega} \subset \mathbb{R}$ is bounded and simply connected with Lipschitz  boundary.
\end{assumption}

\smallskip

\noindent In the sequel we will decompose the vector $x=(x',x_{3})\in\mathbb{R}^{3}$. Let us denote by 
\begin{align}
\delta(x') \  := \ \dist{x'}{\partial \omega},
\end{align}
the distance function between a given $x'\in \omega $ and $\partial \omega$. The in-radius of $\omega$ is then given by 
$$
\mathrm{R}{(\omega)}:= \sup\limits_{x'\in \omega} {\delta(x')}\, .
$$

\subsubsection*{\bf Hardy inequality} Let $c=c(\omega)$ be defined by
\begin{align}\label{bartoscha}
 c^{-2} :=  \inf\limits_{u \in C_{0}^{\infty}( \omega)} \frac{\int_{ \omega}|\nabla_{x'} u(x')|^{2} \df x' }{\int_{  \omega}| u(x')/  {\delta}(x') |^{2} \df x' }\, , 
\end{align}
where $\nabla_x':=(\partial_{x_1},\partial_{x_2})$. 
Clearly, $c$  is the best constant in Hardy's inequality 
\begin{equation} \label{hardy-x'}
\int_{\omega} \frac{u(x')^2}{{\delta}(x')^2} \df x'  \ \leq \ c^2 \int_{ \omega}|\nabla_{x'} u(x')|^{2} \df x', \qquad u \in C^\infty_0(\omega). 
\end{equation} 

\smallskip

\begin{remark}\label{rem}
Under assumption \ref{assumption} it follows from {\mbox{\normalfont {{{\cite{ancona}}}}}} that
\begin{equation} \label{eq-c}
2 \, \leq\,  c \, \leq \,  4. 
\end{equation} 
The best possible value of $c$ is $c =2$. For a  survey on Hardy inequalities we refer to {\normalfont \cite{kufner}}, \mbox{\normalfont \cite{davies}}.
\end{remark}

\noindent  To continue we define for any $\beta>0$ the set $\omega^\beta$ by
$$ 
\omega^{\beta} \ := \ \left\{x'\in \omega\,  | \,  \delta(x')<\beta\right\}\, .
$$
Finally, we introduce the quantity 
$$
l({\omega}) := (b-a)\inf\limits_{0<\beta\leq \mathrm{R}{(\omega)}} \frac{|\omega^{\beta}|}{\beta}\, .
$$
Now can state the main result of this paper.

\begin{thm}\label{dreieins}
Let $\Omega:= \omega\times(a,b)$ and let $c$ is given by \eqref{bartoscha}. Then
\begin{align} \label{main-eq}
\trace{\mathrm{A}(\Omega)-\lambda }_{-}\leq \max\left\{0,\frac{|\Omega|}{96}\,  \lambda^{3}
- \lambda^{\frac{2c+5}{c+2}}  \frac{ (1+\frac 2c)}{96}\,  l(\omega)^{\frac{2c+2}{c+2}}\, |\Omega|^{-\frac{c}{c+2}}\, (4c+4)^{-\frac{2c+2}{c+2}} \right\}
\end{align}
holds for all $\lambda\geq 0$.  
\end{thm}

\begin{remark}
Note that the order of the remainder term is larger than $\lambda^2$ for any $c>0$. So far the order of the second term in the asymptotic expansion \eqref{asymp} it is not known.
\end{remark}

\begin{remark}
For analogous improvements of the classical spectral estimates for Dirichlet Laplacian on bounded domains we refer to \cite{melas}, \cite{kvw}, \cite{yolcu}, \cite{yolcu2} and references therein. 
\end{remark}

\begin{remark}
Following \cite{melas-type} it can be shown that $l(\omega)$ is strictly positive. In particular, it holds $
l(\omega) \geq (b-a)\mathrm{R}{(\omega)}\pi$.
\end{remark}

\begin{corollary} \label{cor1}
Let $\Omega:= \omega\times(a,b)$. If $\omega$ is convex, then 
\begin{align*}
\trace{\mathrm{A}(\Omega)-\lambda }_{-}\leq \max\left\{0,\frac{1}{96} |\Omega| \lambda^{3}
- \lambda^{2+\frac{1}{4}}  \frac{1}{2^7\cdot3^2\sqrt{3}}  \frac{|{\Omega}|}{\mathrm{R}{(\omega)^{3/2}}}  \right\}
\end{align*}
holds for all $\lambda\geq0$. 
 \end{corollary}
\begin{proof}{
In case that $\omega$ is convex we have $c=2$ in \eqref{bartoscha} , see e.g.~\cite{davies}. 
In addition, $\frac{|\omega^{\beta}|}{\beta} $ is a decreasing function of $\beta$ on $(0,\mathrm{R}{(\omega)}]$, see  \cite[Lemma 4.2]{kovarik}.  Hence we compute
\begin{align*}
  l(\omega) = \frac{|{\Omega}|}{\mathrm{R}{(\omega)}}
\end{align*}
and simplify the constant in Theorem \ref{dreieins}.
}\end{proof}

\begingroup
\renewcommand{\i}{\mathrm{i}}

\section{\bf Preliminaries}

\subsection{\bf Magnetic Dirichlet Laplacian}\label{bartosch}
Let $ \mathrm{P_{k,\mathrm{B}}}$ be the orthogonal projection onto the kth Landau level  $\mathrm{B(2k-1)}$ of the Landau Hamiltonian with constant magnetic field for $\mathrm{B}>0$ in $L^{2}(\mathbb{R}^{2})$ and $\mathrm{k}\in \mathbb{N}$. Denote by $ \mathrm{P_{k,\mathrm{B}}}(x,y)$ the integral kernel of $ \mathrm{P_{k,\mathrm{B}}}$. 
 We will need these well-known characteristics
\begin{align}\label{magnet}
\begin{split}
 \mathrm{P_{k,\mathrm{B}}}(y,y) \ &= \ \frac{1}{2\pi}\mathrm{B}, \qquad \text{where} \ y\in \mathbb{R}^{2}, \\
 \int_{\mathbb{R}^{2}} \left(\int_{\Omega} \left| \mathrm{P_{k,\mathrm{B}}}(x,y) \right|^{2} \df y  \right) \df x \ &= \
  \int_{\Omega} \left(\int_{\mathbb{R}^{2}}   \mathrm{P_{k,\mathrm{B}}}(x,y) {\mathrm{P_{k,\mathrm{B}}}(y,x)} \df x  \right) \df y \\  
\  &=  \  \int_{\Omega} \   {\mathrm{P_{k,\mathrm{B}}}(y,y)} \df y = \frac{\mathrm{B}}{2 \pi} |\Omega|.
\end{split}
\end{align}

\subsection{A boundary estimate for the Heisenberg Laplacian}\label{zweizwei}
In this subsection we will derive a boundary estimate for the operator $\mathrm{A}(\Omega)$ which will be crucial in estimating the size of the remainder term in Theorem \ref{dreieins}.  

\begin{prop}\label{prop}Let $\Omega:=\omega\times(a,b) \subset \mathbb{R}^{3}$ and let $c$ be given by \eqref{bartoscha}. Then
\begin{align}
\int_a^b\int_{\omega^{\beta}} |u(x',x_3)|^{2} \df x' \df x_3 \, \leq \, c^{2+\frac 2c} \beta^{2+\frac 2c}\, \norm{\mathrm{A}(\Omega)\, u}_{L^{2}(\Omega)} \norm{\mathrm{A}(\Omega)^{1/c}\, u}_{L^{2}(\Omega)} 
\end{align}
holds for all $u\in \mathrm{Dom}(\mathrm{A}(\Omega))$ and any $\beta>0$
\end{prop}

\noindent For the proof we need the following Lemma.

\begin{lemma}\label{lemmaaa}
Let $\Omega:=\omega\times(a,b) \subset \mathbb{R}^{3}$. Then for all $u \in d[a]$, the form domain of  the closure of \eqref{einseins},  and any $\beta>0$ we have
\begin{align*}
 \int_{a}^{b}  \int_{\omega} \frac{|u(x',x_{3})|^{2}}{\delta(x')^2} \df x'  \df x_{3} \leq c^2 \,  a[u].
\end{align*}
\end{lemma}
\begin{proof}{Let $u$ be in $C_{0}^{\infty}(\Omega)$. In addition let us denote by $\mathcal{F}_{3}$ the Fourier transform in $x_{3}$-direction, which is a unitary map in $L^{2}(\mathbb{R})$. Because $\Omega$ is a cylinder, the function $\left|\mathcal{F}_{3} u(x',\xi_{3})\right|$ lies in $\mathrm{H}_{0}^{1}(\omega)$ for fixed $\xi_{3}\in \mathbb{R}$.
Therefore we can apply inequality \eqref{hardy-x'} to get
\begin{align*} 
\int_{a}^{b}  \int_{\omega}\frac{|u(x',x_{3})|^{2}}{\delta(x')^2} \df x'  \df x_{3} &=\int_{\mathbb{R}} \int_{\omega} \left(\frac{ |\mathcal{F}_{3} u(x',\xi_{3})|}{ {\delta}(x')}\right)^{2} \df x'  \df \xi_{3} \\
&\leq  c^2 \int_{\mathbb{R}} \int_{\omega}  \left( \nabla_{x'} |\mathcal{F}_{3} u(x',\xi_{3})|  \right)^{2}  \df x'  \df \xi_{3}.
\end{align*}
Now we set 
\begin{equation} \label{vec-pot}
{\mathrm{\textbf{A}}}(x'):=\frac{1}{2}(-x_{2},x_{1}),
\end{equation}
and apply the diamagnetic inequality which states that
\begin{align} \label{diamag}
\left| \nabla |\psi| \right| \ \leq \ \left|(\i\nabla+\mathrm{\textbf{A}}) \psi\right| \qquad \text{a. e.}
\end{align}
holds for all $\psi\in H^1(\omega)$, see e.g.~{\normalfont\cite{loss}}. This gives
\begin{align*} 
\begin{split}
 \int_{\mathbb{R}} \int_{\omega}  \left( \nabla_{x'} |\mathcal{F}_{3} u(x',\xi_{3})|  \right)^{2}  \df x'  \df \xi_{3}
  &\leq \    \int_{\mathbb{R}} \int_{\omega}  \left|\left(\i \nabla_{x'}+\xi_{3}{\mathrm{\textbf{A}}}(x')\right)\mathcal{F}_{3} u(x',\xi_{3}) \right|^{2}  \df x'  \df \xi_{3}.
     \end{split}
\end{align*}
Integration by parts in the $x_{3}$-direction yields the inequality for $u \in C_{0}^{\infty}(\Omega)$. A density argument completes the proof.
}\end{proof}

\begin{proof}[Proof of Proposition {\normalfont\ref{prop}}] We follow the proof of \cite[Thm 4]{davisss}. Let us fix $u\in \mathrm{Dom}(\mathrm{A}(\Omega))$ and set
\begin{align*}
\varphi(x):=(\max\{\delta(x'),\beta\})^{-1/c}.
\end{align*}
for $x:=(x',x_3)\in{\Omega}$ and $\beta>0$. In what follows we will use the notation
$$
\nabla_{\mathbb{H}} = (X_1, X_2)
$$
to denote the Heisenberg gradient. First we check that $\varphi u\in  d[a]$. 
Since  $u\in\mathrm{Dom}(\mathrm{A}(\Omega))\subseteq d[a]$,  $\varphi\in H_0^{1}(\omega)$ and get
\begin{align}\label{finitness}
\int_\Omega |\nabla_{\mathbb{H}}(\varphi(x)u(x))|^2 \df x \leq2 \int_\Omega |\varphi(x)\nabla_{\mathbb{H}}u(x)|^2 \df x+2 \int_\Omega |\nabla_{x'}\varphi(x)|^2|u(x)|^2 \df x.
\end{align}
Note that we used here $\nabla_{\mathbb{H}}\varphi(x)=\nabla_{x'}\varphi(x)$ for all $x\in \Omega$. The Eikonal equation
\begin{align} \label{eikonal}
 |\nabla_{x'}\varphi(x)|^2=1 \quad \text{a.e.} \ x\in\Omega,
\end{align}
and the boundedness of $\Omega$ yield the finiteness of \eqref{finitness}. Hence $\varphi u\in  d[a]$ and we may use Lemma \ref{lemmaaa} to get
\begin{align*}
c^{-2}\int_{\Omega} \frac{|\varphi(x)u(x)|^2}{\delta(x')^2} \df x &\leq   \int_\Omega |\varphi(x)\nabla_{\mathbb{H}}u(x)+u(x)\nabla_{\mathbb{H}}\varphi(x)|^2 \df x\\
&=\sprod{\varphi^2\nabla_{\mathbb{H}}u}{\nabla_{\mathbb{H}}u}+\sprod{u}{|\nabla_{\mathbb{H}}\varphi|^2u} \\ &+\frac{1}{2}\sprod{\nabla_{\mathbb{H}}u}{u\nabla_{\mathbb{H}}(\varphi^2)}+\frac{1}{2}\sprod{u\nabla_{\mathbb{H}}(\varphi^2)}{\nabla_{\mathbb{H}}u},
\end{align*}
where we denote by $\sprod{\cdot}{\cdot}$ the scalar product in $L^{2}(\Omega)$. An integration by parts in the last two terms yields
\begin{align*}
 c^{-2}\int_{\Omega} \frac{|\varphi(x)u(x)|^2}{\delta(x')^2} \df x\,  \leq\,  \Re\sprod{\varphi^2u}{\mathrm{A}(\Omega)u}+\sprod{u}{|\nabla_{\mathbb{H}}\, \varphi|^2u}.
\end{align*}
Next we will  estimate the first term on the right hand side. To this end we use Lemma \ref{lemmaaa}, which gives
\begin{align*}
\delta^{-2} \leq c^2 \mathrm{A}(\Omega)
\end{align*}
in the operator sense. Then, by the Heinz inequality, see \cite[Lemma 4.20]{daviesss}, 
\begin{align*}
\varphi^4\leq (\delta^{-2})^{2/c}\leq  \left(c^{2} \mathrm{A}(\Omega)\right)^{2/c}.
\end{align*}
Since $\mathrm{A}(\Omega)^{-1/c}$ is  bounded in $L^{2}(\Omega)$ we obtain
\begin{align*}
\norm{ \varphi^2 \mathrm{A}(\Omega)^{-1/c} }\leq c^{2/c},
\end{align*}
where $\|\cdot\|$ stands for the operator norm in $L^{2}(\Omega)$. Hence
\begin{align*}
|\sprod{\mathrm{A}(\Omega)u}{\varphi^2u}|=|\sprod{\mathrm{A}(\Omega)u}{\varphi^2\mathrm{A}(\Omega)^{-1/c}\mathrm{A}(\Omega)^{1/c}u}| \leq \norm{\mathrm{A}(\Omega)u}_{L^{2}(\Omega)} c^{2/c} \norm{\mathrm{A}(\Omega)^{1/c}u}_{L^{2}(\Omega)}.
\end{align*}
So we arrive at
\begin{align}\label{endineq}
 c^{-2}\int_{\Omega} \frac{|\varphi(x)u(x)|^2}{\delta(x')^2} \df x \leq  \norm{\mathrm{A}(\Omega)u}_{L^{2}(\Omega)} c^{2/c} \norm{\mathrm{A}(\Omega)^{1/c}u}_{L^{2}(\Omega)}+\sprod{u}{|\nabla_{\mathbb{H}}\varphi|^2u}.
\end{align}
On the other hand, the Eikonal equation \eqref{eikonal} implies that
\begin{align*}
|\nabla_{\mathbb{H}}\varphi(x)|^2=|\nabla_{x'}\varphi(x)|^2=c^{-2}\delta(x')^{-2/c-2}\chi_{\{\delta(x')\geq\beta \}}(x'),
\end{align*}
where $\chi_{\{\delta(x')>\beta \}}$ is the characteristic function of the set $\{x\in {\Omega}\, |\,  \delta(x')\geq\beta\}$. Inserting the above identity into \eqref{endineq} thus yields
\begin{align}
 \int_{\{x\in \Omega | \delta(x')<\beta \}} \frac{u(x)|^2}{\delta(x')^2} \df x \leq \beta^{2/c} \norm{\mathrm{A}(\Omega)u}_{L^{2}(\Omega)} c^{2+2/c} \norm{\mathrm{A}(\Omega)^{1/c}u}_{L^{2}(\Omega)}.
\end{align}
The result now follows from the estimate
\begin{align*}
  \int_{\{x\in \Omega | \delta(x')<\beta \}}  {|u(x)|^2} \df x  \leq \beta^{2}  \int_{\{x\in \Omega | \delta(x')<\beta \}} \frac{u(x)|^2}{\delta(x')^2} \df x.
\end{align*}
\end{proof}

\endgroup

\begingroup
\renewcommand{\i}{\mathrm{i}}

\section{\bf Proof of  Theorem \ref{dreieins} }

\renewcommand{\i}{\mathrm{i}}

\noindent Here and below we write a vector $x\in \R^3$ as $x=(x', x_{3})$.
Let $v_{j}$ the orthonormal basis of the eigenfunctions of $\mathrm{A}(\Omega) $ for $j \in \mathbb{N}$;
\begin{align*}
\mathrm{A}(\Omega) v_{j} \ = \ \lambda_{j}  v_{j}, \quad \norm{v_{j}}_{L^{2}(\Omega)}=1.
\end{align*}
Let $\mathcal{F}_{3}$ be the partial Fourier transform in the $x_3$ variable. Then
\begin{align*}
\mathcal{F}_{3} \mathrm{A}(\R^3)\mathcal{F}_{3}^{*} \ = \ \left(\i \partial_{x_{1}}-\frac{1}{2}x_{2}\xi_{3} \right)^{2}+\left(\i\partial_{x_{2}}+\frac{1}{2}x_{1}\xi_{3} \right)^{2} \ = \ \left(\i \nabla_{x'}+\xi_{3}{\mathrm{\textbf{A}}}(x') \right)^{2},
\end{align*}
where ${\mathrm{\textbf{A}}}(x')$ is given by \eqref{vec-pot}. At this point we use the properties of the magnetic Laplacian, see section \ref{bartosch}, to obtain
\begin{align}\label{decomposition}
\mathcal{F}_{3}\, \mathrm{A}(\R^3)\, u(x',\xi_{3}) = \sum_{\mathrm{k=1}}^{\infty} |\xi_{3}|(2\mathrm{k}-1) \int_{\mathbb{R}^{2}} \mathrm{P}_{\mathrm{k},\xi_{3}}(x',y') \mathcal{F}_{3} u(y',\xi_{3}) \df y'
\end{align}
for $\mathcal{F}_{3}\, u(\cdot, \xi_{3})$ in the domain of the magnetic Laplacian.

\subsection{The sharp leading term}
First of all we extend  for every $j \in \mathbb{N}$ the eigenfunctions by $v_{j}(x):=0$ for all $x \in \Omega^{c}$.
Now we consider 
\begin{align*}
\trace{\mathrm{A}(\Omega)-\lambda}_{-} \ =&  \ \sum_{j: \lam<\lambda} \lambda  \norm{v_{j}}_{L^{2}(\R^{3})}^{2} - \norm{X_1v_{j}}_{L^{2}(\mathbb{R}^{3})}^{2}  - \norm{X_2v_{j}}_{L^{2}(\mathbb{R}^{3})}^{2}   \\
= & \ \int_{\mathbb{R}}   \sum_{j: \lam<\lambda} \lambda  \norm{ \mathcal{F}_{3} v_{j}(\cdot,\xi_{3})}_{L^{2}(\mathbb{R}^{2})}^{2} - \norm{\left(\iii\partial_{x_{1}}-\frac{1}{2}x_{2} \xi_{3} \right) \mathcal{F}_{3} v_{j}(\cdot,\xi_{3})}_{L^{2}(\mathbb{R}^{2})}^{2} \df \xi_{3} \\  
-& \ \int_{\mathbb{R}}\sum_{j: \lam<\lambda}   \norm{\left(\iii\partial_{x_{2}}   +\frac{1}{2}x_{1}\xi_{3}\right) \mathcal{F}_{3} v_{j}(\cdot,\xi_{3})}_{L^{2}(\mathbb{R}^{2})}^{2} \df \xi_{3}.
\end{align*}
At this point we apply the spectral decomposition \eqref{decomposition} of the free Heisenberg Laplacian. An application of Fatou's Lemma then yields 
\begin{align*}
\trace{\mathrm{A}(\Omega)-\lambda}_{-} \ \leq \ \int_{\mathbb{R}}   \sum_{j: \lam<\lambda}  \sum_{\mathrm{k}=1}^{\infty}(\lambda-|\xi_{3}|(2\mathrm{k}-1))_+ \norm{f_{j,\mathrm{k},\xi_{3}}}_{L^{2}(\mathbb{R}^{2})}^{2} 
 \df \xi_{3},
\end{align*}
where 
\begin{align*}
f_{{j,\mathrm{k},\xi_{3}}}(x')  & := \ \int_{\mathbb{R}^{2}} \mathrm{ P_{k, \xi_{3}}} (x',y') \mathcal{F}_{3} v_{j}(y', \xi_{3}) \df y' = \ \frac{1}{\sqrt{2\pi}}  \int_{\Omega} \mathrm{P_{k, \xi_{3}}} (x',y') \exp{-\i y_{3} \xi_{3}}v_{j}(y', y_{3}) \df y'  \df y_{3} \\ &=  \frac{1}{\sqrt{2\pi}}  \sprod{\mathrm{P_{k, \xi_{3}}}(x',\cdot) \exp{-\i \cdot \xi_{3}}}{ v_{j}(\cdot)}_{L^{2}(\Omega)}.
\end{align*}
We split the sum as follows;
\begin{align}
\trace{\mathrm{A}(\Omega)-\lambda }_{-} \ \leq&\ \int_{\mathbb{R}} \sum_{\mathrm{k}: \lambda>|\xi_{3}|(2\mathrm{k}-1)} (\lambda-|\xi_{3}|(2\mathrm{k}-1) )  \sum_{j=1}^{\infty} \norm{f_{j,\mathrm{k},\xi_{3}}}_{L^{2}(\mathbb{R}^{2})}^{2}  \df \xi_{3}   \nonumber \\
-&\int_{\mathbb{R}} \sum_{\mathrm{k}: \lambda>|\xi_{3}|(2k-1)} (\lambda-|\xi_{3}|(2\mathrm{k}-1) )\sum_{j: \lam\geq\lambda}^{\infty} \norm{f_{j,\mathrm{k},\xi_{3}}}_{L^{2}(\mathbb{R}^{2})}^{2}  \df \xi_{3},  \label{sum}
\end{align}
noting that the first term on the right hand side is positive and the other one is negative. The completeness of $v_{j}$ and the traces of $\mathrm{P_{k,\xi_{3}}}$, see \eqref{magnet}, yield
\begin{align} \label{vierzwei}
\sum_{j=1}^{\infty} \norm{f_{j,\mathrm{k},\xi_{3}}}_{L^{2}(\mathbb{R}^{2})}^{2}  =& \ \frac{1}{{2\pi}} \int_{\mathbb{R}^{2}} \sum_{j=1}^{\infty}  \left|  \sprod{\mathrm{P_{k, \xi_{3}}}  (x',\cdot) \exp{-\i \cdot \xi_{3}}}{ v_{j}(\cdot)}_{L^{2}(\Omega)}\right|^{2} \df x' = \  \frac{|\xi_{3}|}{4\pi^{2}}|\Omega|.
\end{align}
To obtain the sharp leading term in \eqref{main-eq} we apply this identity in the first integral on the right hand side of \eqref{sum}. Using the fact that   
$$
\sum\limits_{j=1}^{\infty}\frac{1}{(2j-1)^{2}}= \frac{\pi^{2}}{8},
$$ 
we thus get
\begin{align*}
&\int_{\mathbb{R}} \sum_{\mathrm{k}: \lambda>|\xi_{3}|(2\mathrm{k}-1)}(\lambda-|\xi_{3}|(2\mathrm{k}-1) ) \sum_{j=1}^{\infty} \norm{f_{j,\mathrm{k},\xi_{3}}}_{L^{2}(\mathbb{R}^{2})}^{2}  \df \xi_{3} \\
& = \ \frac{|\Omega|}{4 \pi^{2}}\, \int_{\mathbb{R}} \sum_{\mathrm{k}: \lambda>|\xi_{3}|(2\mathrm{k}-1)}(\lambda-|\xi_{3}|(2\mathrm{k}-1)) |\xi_{3}|   \df \xi_{3} \\
&  =  \ \frac{|\Omega|}{2 \pi^{2}} \, \sum_{\mathrm{k}=1}^{\infty} \frac{1}{(2\mathrm{k}-1)^{2}} \int_{0}^{\infty} s(\lambda-s)_{+}   \df s  \ = \ \frac{|\Omega|}{96} \,  \lambda^{3}.
\end{align*}
Inserting this back into \eqref{sum} gives
\begin{align}\label{vierdrei}
\trace{\mathrm{A}(\Omega)-\lambda }_{-}\leq \frac{ |\Omega|}{96}\, \lambda^{3}
-\int_{\mathbb{R}} \sum_{\mathrm{k}: \lambda>|\xi_{3}|(2k-1)}( \lambda-|\xi_{3}|(2\mathrm{k}-1) ) \sum_{j: \lam\geq\lambda}^{\infty} \norm{f_{j,\mathrm{k},\xi_{3}}}_{L^{2}(\mathbb{R}^{2})}^{2}  \df \xi_{3}.  
\end{align}

\subsection{The lower order term}
To obtain a suitable lower bound on the second term in \eqref{vierdrei} we use the same technique as in \cite{kovarik}. The key point of this approach is to estimate the quantity
\begin{align*}
\mathcal{R}_\lambda  := \sum_{j: \lam\geq\lambda} \norm{f_{j,\mathrm{k},\xi_{3}}}_{L^{2}(\mathbb{R}^{2})}^{2}
\end{align*}
from below by a power function of $\lambda$. Note that 
\begin{align*}
 &\mathcal{R}_\lambda \ = \ \frac{|\xi_{3}|}{4\pi^{2}}|\Omega|-  \sum_{j: \lam<\lambda}\norm{f_{j,\mathrm{k},\xi_{3}}}_{L^{2}(\mathbb{R}^{2})}^{2} = \\
 &\frac{1}{2\pi} \int_{\mathbb{R}^{2}}  \int_{\Omega} \left|\mathrm{P_{k, \xi_{3}}} (x',y')\exp{-\i y_{3}\xi_{3}}- \sum_{j: \lam<\lambda}   \sprod{\mathrm{P_{k, \xi_{3}}}  (x',\cdot) \exp{-\i \cdot \xi_{3}}}{ v_{j}(\cdot)}_{L^{2}(\Omega)} v_{j}(y',y_{3})  \right|^{2}  \df y'\df y_{3}\df x'.
\end{align*}
The inclusion $\omega \supseteq \omega^{\beta}$ and an application of $|z-w|^{2}\geq \frac{1}{2}|z|^{2}-|w|^{2}$, with $z,w \in \mathbb{C}$,  imply that 
\begin{align*}
\mathcal{R}_\lambda \ &\geq \ \frac{|\xi_{3}|}{8\pi^{2}}(b-a)\,  |\omega^{\beta}| \\ 
&- \frac{1}{2\pi} \int_{\mathbb{R}^{2}}  \int_{a}^{b}  \int_{\omega^{\beta}}  \left| \sum_{j: \lam<\lambda}   \sprod{\mathrm{P_{k, \xi_{3}}}  (x',\cdot) \exp{-\i \cdot \xi_{3}}}{ v_{j}(\cdot)}_{L^{2}(\Omega)} v_{j}(y',y_{3})  \right|^{2} \df y' \df y_{3}
 \df x' .
\end{align*}
Next we  estimate the negative integral. 
Note that  the linear combinations of $v_{j}$ lie in $ \mathrm{Dom}(\mathrm{A}(\Omega))$. Therefore we may apply Proposition \ref{prop} and obtain
\begin{align*}
&\frac{1}{2\pi} \int_{\mathbb{R}^{2}}\left( \int_{a}^{b} \int_{\omega^{\beta}} \left|\sum_{j: \lam<\lambda}  \sprod{\mathrm{P_{k, \xi_{3}}}  (x',\cdot) \exp{-\i \cdot \xi_{3}}}{ v_{j}(\cdot)}_{L^{2}(\Omega)} v_{j}(y',y_{3}) \right|^{2} \df y' \df y_{3} \right) \df x' \\
    & \leq \ c^{2+2/c} \beta^{2+2/c}   \lambda^{1+1/c} \frac{1}{2\pi} \int_{\mathbb{R}^{2}}\left(    \sum_{j: \lam<\lambda}  \ \left|    \sprod{P_{k, \xi_{3}} (x',\cdot) \exp{-\i \cdot \xi_{3}}}{ v_{j}(\cdot)}_{L^{2}(\Omega)} \right|^{2} \right) \df x' \\ 
  & \leq \ c^{2+2/c} \beta^{2+2/c}   \lambda^{1+1/c} \frac{1}{2\pi} \int_{\mathbb{R}^{2}}\left(    \int_{\Omega}\left| P_{k,\xi_{3}}(x',y')   \right|^{2} \df y' \df x_{3} \right) \df x' \ \\
 & = \ c^{2+2/c} \beta^{2+2/c}   \lambda^{1+1/c} \frac{|\Omega|}{4\pi^2}|\xi_3|,
\end{align*}
which yields  the following lower bound on $\mathcal{R}_\lambda$:
\begin{align*}
\mathcal{R}_\lambda \ \geq& \ \frac{|\xi_{3}|}{8\pi^{2}}(b-a) \left|\omega^{\beta}\right|-   c^{2+2/c} \beta^{2+2/c}   \lambda^{1+1/c} \frac{|\Omega|}{4\pi^2}|\xi_3|\\
 \geq& \  \frac{|\xi_{3}|}{8\pi^{2}}\beta\left( l(\omega)-   2c^{2+2/c} \beta^{1+2/c}   \lambda^{1+1/c} {|\Omega|}\right).
\end{align*}
Now we  set
$$ 
\beta^{1+2/c}= \frac{l(\omega)}{c^{2+2/c}\lambda^{1+1/c}(4+4/c)|\Omega|}, 
$$
which is possible for $\lambda\geq{\lambda_{1}(\Omega)}$, because of
$$ 
\beta^{1+2/c}\, \leq \, \frac{1}{c^{2+2/c}\lambda_{1}(\Omega)^{1+1/c}(4+4/c) \mathrm{R}{(\omega)}}\, \leq\,   \frac{\mathrm{R}{(\omega)}^{1+2/c}}{4}\, .
$$
The last inequality was obtained by applying Proposition \ref{prop} to $u=v_{1}$ and $\beta= \mathrm{R}{(\omega)}$.
Summing up we thus arrive at 
\begin{align*}
\mathcal{R}_\lambda \geq \  \frac{|\xi_{3}|}{8\pi^{2}}\lambda^{-\frac{c+1}{c+2}} l(\omega)^{\frac{2c+2}{c+2}} |\Omega|^{-\frac{c}{c+2}}    (2+4/c)(4c+4)^{-\frac{2c+2}{c+2}}  = \lambda^{-\frac{c+1}{c+2}}\, G(\Omega)\, |\xi_3|,
\end{align*}
where 
$$
G(\Omega) :=  \frac{l(\omega)^{\frac{2c+2}{c+2}}}{8\pi^{2}}\,  |\Omega|^{-\frac{c}{c+2}}    (2+4/c)(4c+4)^{-\frac{2c+2}{c+2}} \, .
$$
This in combination with \eqref{vierdrei} gives 
\begin{align*}
\trace{\mathrm{A}(\Omega)-\lambda }_{-}\leq \frac{ |\Omega|}{96}\,  \lambda^{3}
-G(\Omega)\lambda^{-\frac{c+1}{c+2}} \int_{\mathbb{R}} \sum_{\mathrm{k}: \lambda>|\xi_{3}|(2k-1)}( \lambda-|\xi_{3}|(2\mathrm{k}-1) )|\xi_3| \df \xi_{3}.  
\end{align*}
To finish the proof we calculate in the same way as in the beginning of the proof:
\begin{align*}
 \sum_{\mathrm{k}=1}^{\infty} \int_{0}^{\infty} (\lambda-\xi_{3}(2\mathrm{k}-1))_{+} {\xi_{3}} \df \xi_{3} \ =  \sum_{\mathrm{k}=1}^{\infty} \frac{1}{(2\mathrm{k}-1)^{2}} \int_{0}^{\infty} s(\lambda-s)_{+}   \df s \ =  \ \frac{\pi^{2}\, \lambda^{3}}{48}\, .
\end{align*}
This gives the estimate stated in Theorem \ref{dreieins}. 
\endgroup

\section*{\bf Aknowledgements}
 H.~K. was supported by the Gruppo Nazionale per Analisi Matematica, la Probabilit\`a e le loro Applicazioni (GNAMPA) of the Istituto Nazionale di Alta Matematica (INdAM).
The support of MIUR-PRIN2010-11 grant for the project  ``Calcolo delle variazioni'' (H.~K.) is also gratefully acknowledged. 

B.~R. was supported by the German Science Foundation through the Research Training Group 1838: \textit{Spectral Theory and Dynamics of Quantum Systems.}

 \bibliographystyle{plain}
\bibliography{heisenberg}


\end{document}